\documentclass[11pt]{article}
\usepackage{amsmath,amssymb,amsthm}
\usepackage[sort,numbers]{natbib}
\usepackage{mathrsfs}
\usepackage{framed}
\usepackage{todonotes}
\usepackage{fullpage}
\usepackage[ruled,vlined]{algorithm2e}
\usepackage{algpseudocode}
\usepackage{tikz}
\usepackage{tkz-graph}
\usepackage{pgfplots}
\usepackage{lmodern}
\usepackage{authblk}

\numberwithin{equation}{section}
\newtheorem{theorem}{Theorem}
\newtheorem{proposition}[theorem]{Proposition}
\newtheorem{example}[theorem]{Example}
\newtheorem{lemma}[theorem]{Lemma}

\newtheorem{definition}[theorem]{Definition}

\newcommand{\C}{\mathbb{C}}

\renewcommand{\P}{\mathbb{P}}

\DeclareMathOperator{\rank}{rank}
\DeclareMathOperator{\reg}{reg}

\DeclareMathOperator{\jac}{jac}
\DeclareMathOperator{\bideg}{bideg}

\newcommand{\N}{\mathbb{N}}
\newcommand{\Q}{\mathbb{Q}}

\newcommand{\conormal}{\mathcal N}

\tikzset{%
  >=latex, % option for nice arrows
}

\def\a {\ensuremath{\mathbf{a}}}

\def\reg {\ensuremath{\mathrm{reg}}}
\def\sing {\ensuremath{\mathrm{sing}}}
\def\grad {\ensuremath{\nabla}}

\def\f{{\mathbf{f}}}
\def\g{{\mathbf{g}}}
\def\h{{\mathbf{h}}}

\def\aff{{{\mathsf{aff}}}}

\def\balpha{\mbox{\boldmath$\alpha$}}

\def\bpsi{\mbox{\boldmath$\psi$}}
\def\bphi{\mbox{\boldmath$\phi$}}

\begin{document}
%%\conferenceinfo{WOODSTOCK}{'97 El Paso, Texas USA}

\title{Critical Point Computations on Smooth Varieties:\\Degree and Complexity
bounds}

\author{Mohab Safey El Din}%\titlenote
\affil{Sorbonne Universit{\'e}s, UPMC Univ Paris 06, 7606, LIP6,
    F-75005, Paris, France\\
CNRS, UMR 7606, LIP6, F-75005, Paris, France\\
Inria, Paris Center, PolSys Project}
\author{Pierre-Jean Spaenlehauer}
\affil{Inria, CNRS, Universit\'e de Lorraine\\
Nancy, France}

\date{}

\maketitle

\begin{abstract} Let $V\subset \C^n$ be an equidimensional algebraic
  set and $g$ be an $n$-variate polynomial with rational coefficients.
  Computing the critical points of the map that evaluates $g$ at the
  points of $V$ is a cornerstone of several algorithms in real
  algebraic geometry and optimization.
  Under the assumption that the critical locus is finite and that the
  projective closure of $V$ is smooth, we provide sharp upper bounds
  on the degree of the critical locus which depend only on $\deg(g)$
  and the degrees of the generic polar varieties associated to $V$.
  Hence, in some special cases where the degrees of the generic polar
  varieties do not reach the worst-case bounds, this implies that the
  number of critical points of the evaluation map of $g$ is less than
  the currently known degree bounds.
  We show that, given a lifting fiber of $V$, a slight variant of an
  algorithm due to Bank, Giusti, Heintz, Lecerf, Matera and Solern\'o
  computes these critical points in time which is quadratic
  in this bound up to logarithmic factors, linear in the complexity of
  evaluating the input system and polynomial in the number of
  variables and the maximum degree of the input polynomials.
\end{abstract}

%\section*{CCS Concepts}
%\textbullet\textbf{Computing methodologies} $\to$ \textbf{Algebraic algorithms; 
%Equation and inequality solving algorithms;} \textit{Symbolic and algebraic algorithms; } 
%\textbullet\textit{Applied computing}
%\keywords{Polynomial system solving; Real geometry; Polynomial optimization}

%%A category including the fourth, optional field follows\dots
%\category{D.2.8}{Software Engineering}{Metrics}[complexity measures,
%performance measures]
% \keywords{Critical points, real geometry, complexity}

\section{Introduction}

\textbf{Problem statement.}  Let
$\f=(f_1,\ldots, f_p)\subset \Q[X_1, \ldots, X_n]$ be a polynomial
system defining a smooth and equidimensional algebraic set
$V\subset\C^n$ of dimension $d$ and $g\in\Q[X_1,\ldots,X_n]$ be a
polynomial of degree $D$. We focus on the complexity of computing the
critical points of the map evaluating $g$ at the points of $V$. These
critical points are defined by $f_1=\cdots=f_p=0$ and by the simultaneous
vanishing of the $(n-d+1)$-minors of the jacobian matrix $\jac(\f, g)$
{\[ \left [
\begin{array}{ccc}
\frac{\partial f_1}{\partial X_1}  &\cdots &\frac{\partial f_1}{\partial X_n} \\
  \vdots & & \vdots \\
\frac{\partial f_p}{\partial X_1}  &\cdots &\frac{\partial f_p}{\partial X_n} \\
\frac{\partial g}{\partial X_1}  &\cdots &\frac{\partial g}{\partial X_n} \\
\end{array}
\right ].
\]} Usually, in symbolic computation, when the critical locus is
finite, we aim at computing a rational parametrization of it, which is
the data $((q, v_1, \ldots, v_n), \lambda)$ where $q$ and the $v_i$'s
lie in $\Q[T]$ ($T$ is a new variable) and $\lambda$ is a linear form
in $X_1, \ldots, X_n$ with $q$ square-free, $\deg(v_i)< \deg(q)$,
$\lambda(v_1, \ldots, v_n)=T\frac{\partial q}{\partial T}\mod q$ and
the set defined by
$$q(\tau)=0, \qquad X_i =\frac{v_i}{{\partial q}/{\partial T}}(\tau)
\quad 1\leq i \leq n$$
coincides with the critical locus under consideration. Observe that
the degree of $q$ coincides with the number of critical points and the
number of rational numbers required in such a rational parametrization
is $O(n\deg(q))$.

Assuming again the critical locus to be finite, several bounds on its
cardinality have been established (see \cite{nie2009algebraic} and
references therein). These bounds depend on $n$, $p$ and the degrees
of $f_1, \ldots, f_p$ and $g$. However, it has been remarked that when
$V$ is not a complete intersection, or when it has some special
properties, the cardinality of the critical locus may be far less than
these bounds and sometimes depends only on $D$ and on some quantities
attached to geometric objects. These latter objects are {\em polar
  varieties} (see \cite{polar97,bank2010geometry}); they may be
understood as the critical loci of the restriction to $V$ of
projections on generic linear subspaces ; we define them further
precisely. 

Assuming the smoothness of the projective closure of $V$ and the
finiteness of the critical locus under study, this paper adresses the
following topical questions. 
\begin{itemize}
\item Provide a bound on the number of complex critical points depending
  on $D$ and on the degrees of the generic polar varieties associated
  to $V$.
\item Find an algorithm computing a rational parametrization of this
  critical locus within an arithmetic complexity which is essentially
  \emph{quadratic} in the obtained bound and polynomial in $p$, $n$, the
  complexity of evaluating $\f,g$ and
  $\max_{1\leq i\leq p}(\deg(f_i))$.
\end{itemize}

\textbf{Motivations and prior works.}  Since local extrema of the
evaluation map of $g$ are reached at critical points, computing
critical points is a basic and useful task for polynomial optimization
(see \emph{e.g.} \cite{GS14,greuet2011deciding,el2008computing}). Because of their topological properties related
to Morse theory, computing critical points is also a subroutine for
many modern algorithms in real algebraic geometry yielding
asymptotically optimal complexities (see \emph{e.g.}
\cite{grigor1988solving,renegar1992computational, heintz1994description} 
and \cite{basu2006algorithms} for a textbook reference on this family
of algorithms). Polar varieties have been introduced in \cite{polar97}
for 
computing sample points in each connected component of a real
algebraic set and this technique has been developed in
\cite{safey2003polar,bank2005generalized};
they are also used for computing roadmaps for deciding connectivity
queries \cite{SaSc11, SaSc13, basu2014baby, BR14}, for computing the
real dimension of a real 
algebraic set (see \cite{bannwarth2015probabilistic} and references
therein) 
or for variant quantifier elimination \cite{hong2012variant}.

Some bounds on the cardinality of the critical locus under
consideration are given in \cite{nie2009algebraic} when
$f_1,\ldots, f_p$ is a regular sequence. These bounds depend on the
degrees of the $f_i$'s, $D$, $n$ and $p$.

Since polynomial systems appearing in applications arise most of the
time with a special structure, a natural question to ask for is to
identify situations where the cardinality of the considered critical
locus is less than what the worst case bounds predicted in
\cite{nie2009algebraic}.

Such situations have been exhibited in \cite{catanese2006maximum}
where critical points are used in computational statistics via the
notion of ML degree.  When $\deg(g)=2$ the cardinality of the critical
locus is bounded by the \emph{generic ED degree} of $V$ which depends
only on the degrees of the generic polar varieties associated to $V$
\cite{draisma2015euclidean}.  These bounds do not require any
smoothness assumption.  The results on polar varieties in
~\cite{piene1978polar, holme1988geometric} play a central role in this
setting.

On the algorithmic side, many recent works have focused on the
complexity of computing critical loci.  The results in
\cite{faugere2012critical, spaenlehauer2014complexity} provide
complexity bounds for computing critical points using Gr\"obner bases
under genericity assumptions on the input polynomials $\f$. The
obtained complexity bounds are not quadratic in the generic number of
critical points and the genericity assumptions are not well-suited to
the situations we are willing to consider.  The results in
\cite{bank2013degeneracy} provide complexity bounds for a
probabilistic algorithm computing degeneracy loci in time quadratic in
an intrinsic quantity called the \emph{system degree}. This work is
strongly related to the algorithmic framework of the solver proposed
in \cite{giusti2001grobner} and to computational aspects of polar
varieties which have been deeply investigated in the last decades, see
\cite{polar97,bank2005generalized,bank2010geometry} and references
therein. We will use a slight variant of \cite{bank2013degeneracy} for
our algorithmic contribution.

\textbf{Main results.}  Under some smoothness assumptions which are
precised below, we prove a bound on the number of complex critical
points of the map $x \in V\to g(x)$ depending on the degree of $g$ and
integers $\delta_1(V),\ldots, \delta_{d+1}(V)$. The number
$\delta_{i+1}(V)$ is the degree of the polar variety of $V$ associated
to a generic linear projection on $\C^{i+1}$.  In the sequel, given
$a=(a_1, \ldots, a_i)\subset \C^{ni}$, we denote by $W((g,a), V)$ the
critical locus of the map
$x \in V\to (g(x), a_1\cdot{} x, \ldots, a_i\cdot{} x)$. We also
denote by $\C[X_1,\ldots, X_n]_{\leq D}$ the set of polynomials in
$\C[X_1,\ldots, X_n]$ of total degree $\leq D$.

\begin{theorem}\label{thm:degreeModPolar}\label{THM:DEGREEMODPOLAR}
  Let $V\subset\C^n$ be a $d$-equidimensional algebraic set whose
  projective closure is smooth and $D\geq 1$ and
  $i\in\{0,\ldots, \penalty-1000 d\}$.
  There exists a Zariski dense subset
  $\Omega_i\subset \C[X_1,\ldots, X_n]_{\leq D}\times \C^{i\times n}$
  such that for any $(g,a)\in\Omega$, the degree of
  $W((g,a), V)\subset\C^n$ is bounded by
$$\deg(W((g, a),V))\leq \begin{cases}
\delta_{i+1}(V) \text{ if $D=1$}\\
\sum_{j=i}^{d} \delta_{j+1}(V) (D-1)^{j-i} \text{ otherwise}.
\end{cases}
$$
\end{theorem}

One of the central ingredient of the proof is an algebraic version of Thom's weak transversality
Theorem. We use a formalism and notation similar to \cite[Sec 4.2]{SaSc13}
which provides a proof of this result using charts and atlases.
We also show that this degree bound holds under milder assumptions on
$g$: it is sufficient to assume that the evaluation map of $g$ has
finitely many critical points on $V$. Let us see how such a bound
behaves on an example.

\begin{example}
Let $V\subset \C^8$ be the set of points $(x_1,\ldots, x_8)$ where the matrix
$$\begin{bmatrix} 1 & x_1 &x_2\\
  x_3&x_4&x_5\\
  x_6&x_7&x_8
\end{bmatrix}$$
has rank $1$. This variety has dimension $4$, degree $6$ and
$(\delta_1,\ldots, \delta_5)\penalty-1000 =(1,4,10,12,6)$. Consider
$g=\sum_{i=1}^8 i x_i^3$. Representing an open subset of $V$ as the
zero locus of a reduced regular sequence of quadratic polynomials
$f_1,\ldots, f_4$, bounds depending on the degrees of
$f_1,\ldots, f_4, g$ (see \emph{e.g.} \cite[Thm.
2.2]{nie2009algebraic}) would give the upper bound $2608$ for the
number of complex critical points. Theorem \ref{thm:degreeModPolar}
and its variant for non-generic objective functions (see
Prop. \ref{prop:nongeneric} below) yield the bound
$241=\delta_1 + 2 \delta_2 + 4\delta_3 + 8\delta_4+ 16\delta_5$.
Computations show that the evaluation map of $g$ restricted to $V$ has
actually exactly $241$ complex critical points on $V$.
\end{example}

A convenient representation of an equidimensional variety $V$ of
dimension $d$ is a \emph{lifting fiber} for $V$ (see
\cite{giusti2001grobner}). Roughly speaking, this lifting fiber
consists in a rational parametrization of the (finite) set of points
in a section of $V$ by a $(n-d)$-dimensional affine plane, together
with a \emph{lifting system} which allows to reconstruct a curve in
the variety by symbolic Newton-Hensel iteration.  Assuming that a
lifting fiber of $V$ has been precomputed and that $\deg(g)\geq 2$, we
use the algorithm proposed in \cite{bank2013degeneracy} to compute the
critical points. Since this algorithm handles the more general case of
quasi-affine varieties, so does the proposed variant. However, our main
complexity results hold only under the assumption that the projective closure of $V$
is smooth and that the evaluation map of $g$ has finitely many
critical points on $V$.  Our second main result is a proof that the
arithmetic complexity is quadratic up to logarithmic factors in the
degree bound from Theorem \ref{thm:degreeModPolar}, polynomial in $n$,
the maximum of the degrees of the lifting system, $\deg(g)$, and the
complexity of evaluating the lifting system and $g$.

\textbf{Organization of the paper.}  Section \ref{sec:prelim}
describes notation and preliminary results used throughout this
paper. Section \ref{sec:geometry} is devoted to the proof of
Theorem~\ref{thm:degreeModPolar}.  It relies on a transversality
result which is proved in Section~\ref{sec:transversality}. Section
\ref{sec:nongeneric} deals with nongeneric objective
functions. Finally, Section \ref{sec:algo} discusses algorithmic
aspects and complexity bounds.

\textbf{Acknowledgments.} {Mohab Safey El Din is member of and
  supported by Institut Universitaire de France.}

%%% Local Variables:
%%% mode: latex
%%% TeX-master: "main"
%%% End:

\section{Preliminaries}
\label{sec:prelim}
\subsection{Notation and basic definitions} 

We refer to \cite{Shafa} and \cite{Eis13} for basic definitions about
algebraic sets and polynomial ideals. Given an algebraic set
$V\subset \C^n$, we denote by $I(V)$ the ideal associated to $V$.
Given $\f=(f_1, \ldots, f_p)$ in $\Q[X_1, \ldots, X_n]$, the set of
their common solutions in $\C^n$ is denoted by $Z(\f)$ and the ideal
generated by $\f$ is denoted by $\langle \f\rangle$.
We say that $\f=(f_1, \ldots, f_p)$ is a reduced sequence when the
ideal $\langle \f\rangle$ generated by $\f$ is radical.

{\bf Tangent spaces, regular and singular points.}  Let
$V\subset \C^n$ be an algebraic set. For $x \in V$, the tangent space
$T_x V$ at $x$ to $V$ is the vector space defined by
$\sum_{i=1}^n\frac{\partial f}{\partial X_i}(x) Y_i=0$ for any
$f \in I(V)$. Also, given a finite set of generators
$\f=(f_1, \ldots, f_p)$ of $I(V)$, $T_x V$ is the kernel of the
jacobian matrix
$\jac(\f)=\left ( \frac{\partial f_i}{\partial X_j}\right )_{\substack{1\leq i
\leq p\\ 1\leq i \leq n}}$. We denote by $N_x V$ the orthogonal complement to $T_x V$.

Assume now that $V$ is $d$-equidimensional.  The set of points
$x\in V$ where $\dim(T_x V)=d$ is the set of regular points of $V$; we
denote it by $\reg(V)$. The subset of singular points $\sing(V)$ is
the complement of $\reg(V)$ in $V$; it has dimension less than $d$.
Observe that given a finite set of generators $\f$ of $I(V)$,
$\jac(\f)$ has rank $n-d$ at all $x \in \reg(V)$. Also, $N_x V$ is
generated by the gradient vectors of the polynomials in $\f$ evaluated
at $x$.
An equidimensional algebraic set $V$ is said to be smooth when
$\sing(V)$ is empty.

{\bf Zariski topology.} The Zariski topology over $\C^n$ is the
topology for which the closed sets are the algebraic sets of
$\C^n$. Let $f\in \C[X_1, \ldots, X_n]$; we denote by
${\cal O}(f)\subset \C^n$ the subset defined by $f\neq 0$; it is a
Zariski open set, which is non-empty when $f$ is not identically
$0$. 
Further, we will prove some properties depending on parameters that
are generically chosen. That means that, in the parameter space,
there exists a non-empty Zariski open set such that the property is satisfied for any choice of
the parameter values in this set.

{\bf Projective varieties.}  We will consider algebraic sets in the
projective space $\P^{n}(\C)$ defined by homogeneous polynomials. In
the sequel, we use the shorthand $\P^n$ for $\P^n(\C)$.  

Let $V\subset \P^n$ be a projective variety. Notions of dimension,
tangent space and regular (resp. singular) spaces extend to projective
varieties. We denote by $\aff(V)\subset \C^{n+1}$ the Zariski closure
of the set $\{(x_0,\ldots, x_n)\subset \C^{n+1}\mid
(x_0:\cdots:x_n)\in V\}$.
The variety $\aff(V)$ is an \emph{affine cone} (for all $x\in \aff(V)$, $\lambda \in \C$
we have $\lambda x \in \aff(V)$).  
By a slight abuse of notation, when
$V$ is an algebraic set of $\C^n$, we also denote by $\aff(V)$ the
affine cone of the projective closure of $V$.
Let now $V'\subset \C^{n+1}$ be an affine cone. Observe that the map
${\sf proj}: (x_0, \ldots, x_n)\in V'\setminus\{\mathbf{0}\}\to
(x_0:\cdots:x_n)\in \P^n$
sends $V'\setminus\{\mathbf{0}\}$ to a projective set. Besides, for a
projective variety $V\subset \P^n$, ${\sf proj}(\aff(V))=V$.
We also consider bi-projective varieties lying in $\P^n \times\P^n$.
The above constructions extend similarly: to any variety
$V\subset\P^n\times\P^n$ can be associated a cone
$\aff(V)\subset \C^{n+1}\times \C^{n+1}$ which is the Zariski closure
of the set of points $(x_0,\ldots, x_n, y_0, \ldots, y_n)$ such that
$((x_0 : \cdots :x_n),(y_0:\cdots:y_n)) \in V$. The map ${\sf proj}$ is
extended in the following way:
${\sf proj}: (x,y)\in (\C^{n+1}\setminus\{\mathbf{0}\})\times
(\C^{n+1}\setminus\{\mathbf 0\})\to (x,y)\in \P^n\times \P^n$.

\subsection{Atlases and transversality}\label{ssec:transverse}

Let $V\subset \C^n$ be a $d$-equidimensional algebraic set of
codimension $c$ and $S\subset V$ be a subset.  Following the
terminology in \cite[Chap. 5]{SaSc13}, an atlas for $(V, S)$ is a
finite sequence $\bpsi=((\h_j, m_j))_{1\leq j \leq \ell}$, with
$\h_j=(h_{j,1}, \ldots, h_{j, c})\subset \C[X_1, \ldots, X_n]$ and
$m_j\in \C[X_1, \ldots,\penalty-1000 X_n]$ such that for all $1\leq j \leq \ell$
the following holds:
\begin{itemize}
  \item[${\sf P}_1$]
  ${\cal O}(m_j)\cap (V\setminus S)={\cal O}(m_j)\cap
  (Z(\h_j)\setminus S)$;
\item[${\sf P}_2$] ${\cal O}(m_j)\cap (V\setminus S)$ is not empty;
\item[${\sf P}_3$] for all $x \in {\cal O}(m_j)\cap V\setminus S$,
  $\jac(\h_j)$ has full rank $c$ at $x$;
\item[${\sf P}_4$] the open sets ${\cal O}(m_j)$ cover $V\setminus S$.
\end{itemize}
We say that $\h_j$ is a set of local equations over ${\cal O}(m_j)$.
\cite[Lemma 5.2.4]{SaSc13} establishes that there exists an atlas for
$(V, \sing(V))$. Also, observe that
$\sing(V)\subset Z(m_1\cdots m_\ell)$.

Further, we use the notion of transverse intersection for algebraic
sets and projective varieties.  Let $V$ and $W$ be equidimensional
algebraic sets in $\C^n$. As in \cite[pp. 21]{EisenbudHarris}, we say
that $V$ and $W$ intersect transversely at $x$ if
$x \in \reg(V)\cap \reg(W)$ and $T_x V+T_x W = \C^n$.  They intersect
generically transversely if they meet transversely at a generic point
of each irreducible component of $V\cap W$. This definition is
naturally extended to projective varieties.

We say that two sets $V$ and $W$ intersect transversely over an open
set $U$ if $V$ and $W$ intersect tranversely at any point
   of $V\cap W\cap U$.

\begin{lemma}\label{lemma:transverse}
  Let $V_1$ and $V_2$ be equidimensional algebraic sets of
  codimensions $c_1$ and $c_2$.  Consider atlases
  $\balpha_1=((\h_j, m_j))_{1\leq j \leq \ell}$ and
  $\balpha_2=((\g_j, n_j))_{1\leq j \leq k}$ for $(V_1, \sing(V_1))$
  and $(V_2, \sing(V_2))$.  Assume that $V_1\cap V_2$ is either empty
  or that for any irreducible component $Z$ of $V_1\cap V_2$, there
  exist $r\in\{1,\ldots, \ell\}$ and $s\in\{1,\ldots, k\}$ such that
  \begin{itemize}
  \item[${\sf T}_1$] 
    $Z\cap {\cal O}(m_{r}n_{s})$ is not empty;
  \item[${\sf T_2}$] At any point of $\reg(Z)\cap {\cal O}(m_{r}n_{s})$,
    the matrix $\jac(\h_r, \g_s)$ has rank $c_1+c_2$.
  \end{itemize}
  Then $V_1$ and $V_2$ intersect generically transversely. 
\end{lemma}

\begin{proof}
  The equality
  $\rank(\jac(\h_r,\g_s))=\rank(\jac(\h_r))+\rank(\jac(\g_s))$ implies
  that at any point $x\in \reg(Z)\cap {\cal O}(m_{r}n_{s})$,
  $N_x V_1\cap N_x V_2 = 0$. Consequently,
  $T_x V_1 + T_x V_2 =(N_x V_1\cap N_x V_2)^\perp = \C^n$. Finally,
  noticing that $\reg(Z)\cap {\cal O}(m_{r}n_{s})$ is dense in
  $Z\cap {\cal O}(m_{r}n_{s})$, which is dense in $Z$ (by
  ${\sf T}_1$) concludes the proof.
\end{proof}

We also need to prove that the intersection of bi-projective varieties
is transverse. This is done via their associated affine cones.
In the sequel, the set $\{(x, y)\mid
  x=\mathbf{0}\}\subset\C^{n+1}\times\C^{n+1}$ is denoted by
$\mathscr{X}$ and the set $\{(x, y)\mid y=\mathbf{0}\}\subset\C^{n+1}\times\C^{n+1}$ is denoted by
$\mathscr{Y}$.

\begin{lemma}\label{lemma:XY} Let $V_1$ and $V_2$ be projective varieties in
  $\P^n\times\P^n$.  Then $V_1$ and $V_2$ intersect transversely
  at every point $(x,y)=((x_0:\ldots: x_n),(y_0:\ldots: y_n))\in\P^n\times\P^n$
  iff $\aff(V_1)$ and $\aff(V_2)$ intersect transversely over
  $\C^{n+1}\times\C^{n+1}\setminus (\mathscr{X}\cup \mathscr{Y})$.
\end{lemma}

\begin{proof}
  Let $i,j$ be such that $x_i\ne 0$ and $y_j\ne 0$. W.l.o.g., we
  assume that $i=j=0$. Consider the affine chart
  $U\subset\P^n\times \P^n$ defined by $x_0\ne 0, y_0\ne 0$. Let
  $H_1\subset\C^{n+1}\times\C^{n+1}$ (resp. $H_2$) be the hyperplane
  defined by $x_0=1$ (resp. $y_0=1$). For $\ell\in\{1,2\}$, the
  variety $V_\ell\cap U$ can be identified to
  $\aff(V_\ell)\cap H_1\cap H_2$. By definition of transversality, the
  varieties $V_1$ and $V_2$ intersect transversely at
  $(x,y)\in\P^n\times\P^n$ if and only if so do $V_1\cap U$ and
  $V_2\cap U$. By the previous identification, this is equivalent to
  saying that $\aff(V_1)\cap H_1\cap H_2$ and
  $\aff(V_2)\cap H_1\cap H_2$ intersect transversely at
  $(1, x_1,\ldots, x_n, 1, y_1,\ldots, y_n)$. Finally, direct tangent
  space computations show that for $z_1,z_2$ in $\C\setminus \{0\}$
  and for $\ell\in\{1,2\}$,
  $T_{(z_1, z_1 x_1 ,\ldots, z_1 x_n, z_2, z_2 y_1,\ldots, z_2 y_n)}
  \aff(V_\ell) = T_{(1, x_1 ,\ldots, x_n,1,y_1,\ldots y_n)}
  \aff(V_\ell)$.
\end{proof}
\subsection{Critical points and modified polar varieties}

Let $V\subset \C^n$ be an equidimensional algebraic set of
codimension $c$ and $g\in \Q[X_1, \ldots, X_n]$. Consider the
evaluation map $\varphi_g: x \in V \to g(x)$. We denote by
$w(\varphi_g, V)$ the set
$\{x \in \reg(V)\mid {\rm rank}(\jac_x(\mathbf{f}, g))< c+1\}$.  This
is a locally closed constructible set and it coincides with the
critical locus of the map $\varphi_g$.  Its Zariski closure is denoted
by $W(\varphi_g, V)$. This construction can be generalized as follows.

Let $a_1, \ldots, a_n$ be linearly independent vectors in $\C^{n}$ and
for $1\leq i \leq n$, set $\a_i = (a_1, \ldots, a_i)\in\C^{i\times n}$. 
Then,
for $1\leq i \leq n$, let $W( (g,\mathbf a_i), V)$ denote the algebraic
set
$$\{x \in V\mid {\rm rank}(\jac_x(\mathbf{f}, g,\varphi_{\mathbf a_i}))<
c+i+1\},$$
where $\mathbf f$ is a set of generators of $I(V)$, and
$\varphi_{\mathbf a_i}$ is the set of linear forms
$(a_j\cdot X)_{j\in\{1,\ldots, i\}}$ (with $X=(X_1,\ldots, X_n)$).
Reusing the terminology of \cite{GS14}, we call these
sets {\em modified polar varieties} associated to $g$ and $V$, the $i$-th one being
$W((g, \a_i), V)$.
We let $W(\mathbf a_i, V)$
be the classical polar variety
$\{x \in V\mid {\rm rank}(\jac_x(\mathbf{f}, \varphi_{\a_{i}})<
c+i\}$, reusing the letter $W$ for the sake of simplicity.

\begin{proposition}\label{prop:genpolvar}
  Let $V$ be a $d$-equidimensional algebraic set, and $i\in\{1,\ldots, d\}$.  There exists a
  Zariski dense subset $\mathscr{O}\subset \C^{i\,n}$ and an integer
  integer numbers $\delta_i$ such that for any
  $a\in \mathscr{O}$, the following holds:
  \begin{itemize}
  \item $W(a, V)$ is either empty or equidimensional of dimension
    $i-1$;
  \item $W(a, V)$ has degree at most $\delta_i$.
  \end{itemize}
\end{proposition}
\begin{proof}
  The first statement follows directly from \cite[Prop.3]{bank2013degeneracy}.
  For the second statement, we refer to the definition of $\delta_{classic}$ in
  \cite[Sec. 4]{bank2010geometry}.
\end{proof}
The integers $\delta_i$ are denoted by $\delta_i(V)$ in the sequel. By
convention, we set $\delta_{d+1} = \deg(V)$. These numbers are also
called \emph{projective characters} of $V$ (see \cite[Example
14.3.3]{fulton2012intersection}).

%%% Local Variables:
%%% mode: latex
%%% TeX-master: "main"
%%% End:

\section{Proof of Theorem~\ref{thm:degreeModPolar}}
\label{sec:geometry}
We start by introducing some objects which play a central role in the
proof. As before, $V$ is a $d$-equidimensional algebraic set and
$\aff(V)$ denotes the affine cone over the projective closure of $V$.

Let $\conormal_{V}\subset \C^{n+1}\times\C^{n+1}$ be the Zariski
closure of the set
\[\{(x, y)\in \C^{n+1}\times\C^{n+1}\mid x\in
  \aff(V)\setminus\{\mathbf 0\},\, y\in N_x \aff(V)\setminus\{\mathbf
  0\}\}.\] It is called the conormal variety of $\aff(V)$.
  Consider $a=(a_1, \ldots,\penalty-1000 a_i)\in \C^{(n+1)i}$, a homogeneous
polynomial ${g}\in \C[X_0, X_1, \ldots, X_n]$ of degree $D$ and the matrix 
{\begin{equation*}\Sigma_i(g,a)=\begin{bmatrix}
      Y_0&\cdots & Y_n \\
\partial {g}/\partial X_0 & \cdots&\partial {g}/\partial X_n\\
\rule[2pt]{1cm}{.5pt}& a_1 & \rule[2pt]{1cm}{.5pt}\\
&\vdots&\\
\rule[1pt]{1cm}{.5pt} & a_{i}& \rule[2pt]{1cm}{.5pt}\\
\end{bmatrix}.\end{equation*}} Let
$S_i({g},a)\subset\C^{n+1}\times\C^{n+1}$ be the variety defined by
the rank condition $\rank(\Sigma_i({g},a))\leq i+1$.
Let $\Pi$ be the projection
\begin{equation}\begin{array}{cccl}
  \label{eq:pi}
\Pi:&(x,y)\in \C^{n+1}\times\C^{n+1}&\rightarrow&x\in\C^{n+1}.\\
    \end{array}\end{equation}      
  If $a$ is generic, then $\conormal_V\cap S_i({g},a)$ is the
  Zariski closure of set of points
  $(x,y)$ such that $y\in N_x \aff(V)$, and
  $(y_0,\ldots, y_n)\in {\rm Span}(a)+\grad_x {g}$. 
  In other words, $(x_0, x_1,\ldots, x_n)$ is a critical point of the
  map
  $(X_0,\ldots, X_n)\mapsto ({g}(X), a_1\cdot X, \ldots, a_i\cdot X)$.
  Let $a=(a_1',\ldots, a_{i-1}')\in\C^{n(i-1)}$ be a basis of the
  vector space
  $\{(u_1,\ldots,u_n)\in\C^n \mid (0, u_1,\ldots, u_n)\in {\rm
    Span}(a_1,\ldots, a_i)\}$.
  Therefore, if the first coordinate of $a_1$ is non\-zero, then the
  restriction of $\Pi(\conormal_{V}\cap S_i({g},a))$ to the chart
  $x_0=1$ is the modified polar variety $W( (g_{\mid x_0=1}, a'), V)$.
  The set of homogeneous polynomials in $\C[X_0, \ldots, X_n]$ of
  degree $D$ is a finite dimensional vector space; we denote by $N$
  its dimension and identify those homogeneous polynomials to points
  in $\C^N$.  Assume for the moment the following result which is
  proved in Section~\ref{sec:transversality}.
  \begin{proposition}\label{prop:transversality}\label{PROP:TRANSVERSALITY}
    Let $V\subset \C^n$ be a $d$-equi\-dimen\-sional algebraic set
    such that its projective closure is smooth and $i\in\{0,\ldots,
    \penalty-1000 d\}$.
    There exists a non-empty Zariski open set
    $\mathscr{O}\subset \C^{(n+1)i}\times \C^N$ such that for any
    $(g,a)\in \mathscr{O}$, $\conormal_{V}$ and $S_i(g,a)$ meet
    generically transversely over
    $(\C^{n+1}\setminus\{\mathbf{0}\})\times(\C^{n+1}\setminus\{\mathbf{0}\})$.
  \end{proposition}

  One can associate to any equidimensional variety
  $Z\subset\P^n\times\P^n$ of codimension $c$ a bivariate homogenous
  polynomial $\bideg(Z)\in\N[T,U]$ of degree $c$, called the bidegree
  of $Z$ \cite{waerden29, waerden78}. The coefficient of $T^k U^{c-k}$
  in $\bideg(Z)$ is the number of points (counted with multiplicity)
  of $Z\cap (H_1\times \P^n)\cap (\P^n\times H_2)$ where $H_1$
  (resp. $H_2$) is a generic linear space of dimension $n-k$
  (resp. $n-c+k$). 

  By \cite[Sec. 5]{draisma2015euclidean}, the bidegree of
  $\conormal_{V}$ is $\sum_{k=0}^{d}\delta_{k+1}(V)T^{n-k} U^{k+1}$.

  We focus now on the bidegree of $S_i(g,a)$.  

  \begin{lemma}\label{lemma:Si}
    There exists a non-empty Zariski open set
    $\mathscr O'\subset \C^N\times \C^{(n+1)i}$ such that for
    $(g,a)\in \mathscr O'$,
    $S_i(g,a)\subset \C^{n+1}\times\C^{n+1}$
    has codimension $n-i$ and its
    bidegree is
    $\sum_{k=0}^{n-i} (D-1)^{k} T^{k} U^{n-k-i}.$ Moreover,
    $\reg(S_i(g,a))$ coincides with the set of points where the matrix
    $\Sigma_i(g,a)$ has rank $i+1$.
  \end{lemma}

  \begin{proof}
    $S_i(g,a)$ is the variety of $(x,y)\in\C^{n+1}\times\C^{n+1}$
    where the evaluation of $\Sigma_i(g, a)$ is rank defective.  There
    exists a Zariski dense subset $\mathscr O_1\subset \C^{(n+1)i}$
    such that for all $a\in\mathscr O_1$, the top-left
    $i\times i$ submatrix of $A$ is invertible, where $A$ is the
    matrix with rows $a_1,\ldots, a_i$.  For
    $a\in\mathscr O_1$, let $B=(b_{i,j})$ be an invertible
    $(n+1)\times (n+1)$ matrix such that $A\cdot B = [ \mathbf 0 \mid I_i ]$.
    The rank condition on $A\cdot B$ shows that $S_i(g,a)$ is
    the set of points $(x,y)\in\C^{n+1}\times\C^{n+1}$ where the rank of
$$M = \begin{bmatrix}
  \sum_{j=1}^{n+1} b_{j,1} Y_j & \cdots&  \sum_{j=1}^{n+1}
  b_{j,n+1-i} Y_j\\
  \sum_{j=1}^{n+1} b_{j,1} \partial g/\partial X_j&   \cdots&
  \sum_{j=1}^{n+1}
  b_{j,n+1-i} \partial g/\partial X_j
\end{bmatrix}$$
is at most $1$, where $Y_1,\ldots, Y_{n+1-1}$ are new variables.
Next, let $S'_i\subset \P^{n-i}\times\P^{n-i}$ denote the determinantal variety of rank-defective matrices
$$\begin{bmatrix}
  \mathfrak u_{1,0}&\cdots&\mathfrak u_{1, n-i}\\
  \mathfrak u_{2,0}&\cdots&\mathfrak u_{2, n-i}
\end{bmatrix},
$$
together with the grading given by $\deg(\mathfrak u_{1,j})=1$,
$\deg(\mathfrak u_{2,j}) = D-1$ for all $j\in\{0,\ldots, n-i\}$.
Setting $\mathbf s = 0$, $\deg(t_1) = D-1$, $\deg(t_2) = 1$, in
\cite[Example 15.39]{MilStu05}, the multidegree of $S_i'$ is
$\sum_{k=0}^{n-i} (D-1)^{k} T^{k} U^{n-k-i},$ where $T$ (resp. $U$)
corresponds to the class of a hyperplane in the first (resp. second)
operand in the product $\P^{n-1}\times\P^{n-1}$. Let
$\C[X_0,\ldots, X_n]_{D-1}$ denote the set of homogeneous polynomials
of degree $D-1$. Since determinantal varieties are Cohen-Macaulay
\cite[Thm. 11]{northcott1960semi}, by the same argument as in
\cite[Sec. 4]{FSS13}, there exists a Zariski dense subset
$\mathcal O_2\subset \C[X_0,\ldots, X_n]_{D-1}^{n+1}\times
\C^{(n+1)i}$
such that for any $(h_0,\ldots, h_{n}, a)\in \mathscr{O}_2$, the
bidegree of the variety defined by $\rank(M)\leq 1$ also has bidegree
$\sum_{k=0}^{n-i} (D-1)^{k} T^{k} U^{n-k-i}$.  Note that the set of
$(h_0,\ldots, h_n)$ which are of the form
$(\partial g/\partial X_0,\ldots, \partial g/\partial X_n)$ is a
linear subspace of $\C[X_0,\ldots, \penalty-1000 X_n]$.  It remains to
prove that the restriction of $\mathscr{O}_2$ to this subspace is
nonempty. This is done by considering
$(h_0,\ldots, h_n) = (X_0^{D-1}, \ldots,\penalty-1000 X_n^{D-1})$
(which comes from the derivatives of $g=(X_0^D+\cdots+X_n^D)/D$) and
$a_{i,j}=1$ if $i=j$ and $0$ otherwise.  Direct computations show that
the corresponding variety has the expected bidegree.  Therefore the
open set $\mathcal O'$ of pairs $(g, a)$ such that
$(\partial g/\partial X_0,\ldots, \partial g/\partial X_n, a)\in
\mathscr{O}_2$
satisfies the desired properties.  Writing the equations defining the
variety of $S_i(g,a)$ from the rank of the matrix $M$ shows that
$(x,y)\in\sing(S_i(g,a))$ iff the evaluation of the first row of $M$
is zero, which is equivalent to saying that $(Y_0,\ldots, Y_n)$ lies
in ${\rm Span}(a)$. This implies that the regular locus of $S_i$ is
the set of points where $\Sigma_i(g,a)$ has rank $i+1$.
\end{proof}
By Proposition~\ref{prop:transversality}, there exists a non-empty
Zariski open set $\mathscr O\subset \C^N\times \C^{(n+1)(i+1)}$ such
that for $g,a'$ in $\mathscr O$, $\conormal_{V}$ and $S_{i+1}(g,a')$
meet generically transversely outside the set
$\mathscr{X}\cup\mathscr{Y}$ introduced before Lemma~\ref{lemma:XY}.
Consider the map ${\sf proj}$ introduced in Section~\ref{sec:prelim}
(paragraph on projective varieties).  We deduce that for
$(g,a')\in \mathscr O_2$, $\conormal'_V={\sf proj}(\conormal_{V})$ and
$S'_i(g,a')={\sf proj}(S_i(g,a'))$ meet generically transversely
(Lemma~\ref{lemma:XY}). Below, we take
$(g,a)\in \mathscr{O}\cap \mathscr{O}'$ (where $\mathscr{O}'$ is the
non-empty Zariski open set defined in Lemma~\ref{lemma:Si}).

Intersection theory \cite[Theorem Definition 1.7]{EisenbudHarris}
states that if two subvarieties $Z_1$ and $Z_2$ of $\P^n\times \P^n$
intersect generically transversely, then 
$$\bideg(Z_1\cap Z_2) = \bideg(Z_1)\cdot\bideg(Z_2) \bmod \langle T^{n+1},
U^{n+1}\rangle.$$
We deduce that $\bideg(\conormal'_{V}\cap S'_i(g,a))$ equals
{\begin{equation*}
  \begin{split}
    \left(\sum_{k=0}^{d}\delta_{k+1}(V)T^{n-k} U^{k+1}\right)\left(
      \sum_{k=0}^{n-i-1} (D-1)^{k} T^{k} U^{n-k-i-1}\right) \\\bmod
    \langle T^{n+1}, U^{n+1}\rangle.
  \end{split}    
  \end{equation*}}
  Note that the degree of the image of $S'_{i+1}(g,a')\cap\conormal'_{V}$ by
  the projection $\pi_1:(x,y)\mapsto x$ is the coefficient of
  $T^{n-i-1}U^n$ in its bidegree. Direct computations show that it
  equals
   $$
   \begin{cases}
     \delta_{i+1}(V) \text{ if $d=1$}\\
     \sum_{j=i}^{d} \delta_{j+1}(V) (D-1)^{j-i} \text{ otherwise}.
   \end{cases}
   $$

   For $j\in\{1,\ldots, i+1\}$, let $\nu_j$ be the first coefficient
   of $a_j'$ and let $\mathcal U$ be the set of $a'\in\C^{(n+1)(i+1)}$
   such that $\nu_1\ne 0$. Set
   $\tilde{\mathscr{O}} = \{(g,a')\in \mathscr{O}\cap\mathscr{O}' \mid
   a'\in\mathcal U\}$.
   For $a'\in\mathcal U$, let $\chi$ be the map sending $a'$ to
   $(a_2'-\nu_2 a_1'/\nu_1,\ldots, a_i'-\nu_i a_1'/\nu_1).$ The image
   of $\mathcal U$ by $\chi$ is a dense open subset
   $\mathcal U'\subset \C^{ni}$.  Finally, we write $\Omega$ for the
   set
   $(g_{\mid X_0=1}, a)\in\C[X_1,\ldots, X_n]_{\leq D}\times \C^{ni}$
   such that there exists $(g, a')\in \tilde{\mathscr O}$ with
   $\chi(a')=a$.  For $(g_{\mid X_0=1},a)\in\Omega$, $S_{i+1}'(g, a')$
   and $\conormal'_{V}$ intersect generically transversely. Moreover,
   its image by the projection $\Pi$ (see \eqref{eq:pi}) restricted to
   the chart $x_0=1$, $y_0=1$ is $W(g_{\mid X_0=1},a)$.  Consequently,
   \begin{eqnarray*}
     \deg(W(g_{\mid X_0=1},a))&\leq& \deg(\Pi(S_{i+1}'(g^h, a')\cap
   \conormal'_V))\\&=&
   \begin{cases}
     \delta_{i+1}(V) \text{\quad\quad\quad if $d=1$}\\
     \sum_{j=i}^{d} \delta_{j+1}(V) (D-1)^{j-i} \text{ otherwise}.
   \end{cases}
 \end{eqnarray*}

   \def\clos{{V^h}}
\def\Vaff{{V^{(a)}}}

\def\setS{{S^a}}
\def\X{{\mathbf{X}}}
\def\opensetS1{{\Omega_1}}
\section{Proof of Proposition~\ref{prop:transversality}}
\label{sec:transversality}
Our proof relies on applying Lemma~\ref{lemma:transverse} with
$V_1=\conormal_{V}$ and $V_2=S_i(g,a)$ for a generic choice of
$(a, g)$.  It simply consists in proving that properties ${\sf T}_1$
and ${\sf T}_2$ defined in Lemma~\ref{lemma:transverse} hold.
This leads us to define atlases and local equations for $\conormal_V$.
Next, we define an atlas (and hence local equations) for a set related
to $S_i(g,a)$. We will apply an algebraic version of Thom's weak
transversality Theorem to a well chosen map constructed using these
local equations, establishing that this map is regular at the
origin. Finally, we will use these results in the last paragraph of
this Section to prove properties ${\sf T}_1$ and ${\sf T}_2$ under
some genericity assumption on $(g,a)$.

\subsection{Local equations for $\conormal_V$}
By assumption, $V$ is $d$-equidi\-men\-sional and smooth as is its
projective closure; we denote by $c$ its codimension. This implies
that the affine cone $\aff(V)$ of the projective closure of $V$ is
also equidimensional of codimension $c$. Besides, if
$(x, y)\in \aff(V)$ with $x\neq 0$ then $x$ is a regular point of
$\aff(V)$. By \cite[Lemma 5.2.4]{SaSc13}, there exists an atlas
$\bpsi=((\h_j, m_j))_{1\leq j \leq J}$ for $(\aff(V), \sing(\aff(V)))$
(see Subsection~\ref{ssec:transverse}).  This leads us to define the
set
\[
U_{j}=\{(x, y)\mid x \in \aff(V)\cap {\cal O}(m_j), y\perp T_x
\aff(V)\setminus\{\mathbf{0}\}\}.
\]
Since the open sets ${\cal O}(m_j)$ cover $\aff(V)\setminus\sing(\aff(V))$
(property ${\sf P}_4$), the sets $U_j$ cover
$\conormal_V\setminus \mathscr{X}\cup \mathscr{Y}$.

Let $m'_{j,1}, \ldots, m'_{j,L_j}$ be the $c\times c$ minors of $\jac(\h_j)$
such that ${\cal O}(m_jm'_{j,k})\cap V\neq \emptyset$ for
$1\leq k \leq L_j$. For $1\leq r\leq n-c$, we denote by $M_{r,k}(m'_{j,k})$
the minor of the $(c+1, c+1)$ minors of the $(c+1, c+1)$ submatrix of
\[
J=\begin{bmatrix}
  \jac(\h_j) \\
  Y_0 \cdots Y_n \\
\end{bmatrix}
\]
whose upper left $(c\times c)$ minor is $m'_{j,k}$ and adding the
missing row and column. 
In the sequel, we denote by $\mathbf{H}_{j, k}$ the sequence
$\h_j, M_{1,k}(m'_{j,k}), \ldots, M_{n-c, k}(m'_{j,k})$.

\begin{lemma}\label{lemma:atlas:conormal}
  Under the above notation and assumptions, the sequence of couples
  $(\mathbf{H}_{j, k}, m_jm'_{j,k})$ for $1\leq j \leq J$ and
  $1\leq k\leq L_j$ is an atlas for
  $(\conormal_V, \sing(\conormal_V)\cup\mathscr{X}\cup \mathscr{Y})$.
\end{lemma}

\begin{proof}
  Recall that we are given an atlas
  $\bpsi=((\h_j, m_j))_{1\leq j \leq J}$ for
  $(\aff(V), \sing(\aff(V)))$.  Let
  $(x, y)\in
  \conormal_V\setminus(\sing(\conormal_V)\cup\mathscr{X}\cup\mathscr{Y})$.
  Then, $x \in \reg(\aff(V))$ (because $x\neq \mathbf{0}$ and the
  projective closure of $V$ is assumed to be smooth) and there exists
  $1\leq j \leq J$ such that $x\in \aff(V)\cap {\cal O}(m_j)$. Besides
  note that $\aff(V)\cap {\cal O}(m_j)$ coincides with
  $Z(\h_j)\cap {\cal O}(m_j)$ (property ${\sf P}_1$) and that $\jac(\h_j)$
  has maximal rank at $x$ (property ${\sf P}_3$). We let $m'_{j,k}$ be a
  $(c\times c)$-minor of $\jac(\h_j)$ which does not vanish at $x$.

  Since $(x, y)\in \conormal_V$, we have $y \perp T_x \aff(V)$. Using
  property ${\sf P}_1$ and ${\sf P}_3$, we deduce that $T_x \aff(V)$ is the
  kernel of $\jac(\h_j)$. We deduce by elementary linear algebra that
  the matrix $J$ introduced above is rank defective at $(x,
  y)$.
  Besides, elementary linear algebra (e.g. using a Schur complement)
  shows that over ${\cal O}(m_j m'_{j,k})$, the variety defined by
  $\h_j(x)=0$ and ${\sf rank}(J(x, y))\leq c$ is defined by
  $\mathbf{H}_{j,k}$. We have established properties ${\sf P}_1$ and
  ${\sf P}_2$. Establishing the fact that the sets ${\cal O}(m_jm'_k)$
  cover
  $\conormal_V\setminus(\sing(\conormal_V)\cup\mathscr{X}\cup\mathscr{Y})$
  (property ${\sf P}_4$) is immediate from the above discussion.
  It remains to prove that $\jac(\mathbf{H}_{j,k})$ has maximal rank
  at $(x, y)$ (property ${\sf P}_3$). Without loss of generality, assume
  that $m'_{j,k}$ is the upper left minor of $\jac(\h_j)$. Observe
  that the minors $M_{1, k}(m'_{j,k}), \ldots, \penalty-1000 M_{n-c, k}(m'_{j, k})$
  can be written as $Y_{c+\ell}m'_{j,k} +\rho_\ell$ where
  $\rho_\ell\subset \Q[X_1, \penalty-1000\ldots, X_n, Y_1, \ldots, Y_c]$.
  Extracting from $\jac(\mathbf{H}_{j,k})$ the columns of $\jac(\h_j)$
  corresponding to $m'_{j,k}$ and those corresponding to the partial
  derivatives w.r.t $Y_{c+\ell}$ for $1\leq \ell\leq n-c$ yields a
  submatrix which is not rank defective over
  $Z(\h_j)\cap {\cal O}(m_jm'_{j,k})$ which ends the proof.
\end{proof}

\subsection{Local equations for $S_i(g,a)$} 
In this section, we build an atlas for $S_i(g,a)$ for generic
$(g, a)$. To do that, we see $(g,a)$ as in point in the space
$\C^N\times \C^{(n+1)i}$ (recall that $N$ is the dimension of the vector
space of homogeneous polynomials in $\C[X_0, \ldots, X_n]$) and see
the entries of $a$ and the coefficients of $g$ as variables.

Formally, for $1\leq r \leq i$, let $A_r=(A_{0,r}, \ldots, A_{n, r})$
be a vector of indeterminates. Let also
${\cal M}=\{(\alpha_0, \ldots, \alpha_n)\in \N^{n+1}\mid
\sum_{j=0}^n\alpha_j= D\}$
and $G=(G_\alpha, \alpha \in {\cal M})$ be a vector of indeterminates.
By abuse of notation, we also denote by $G$ the polynomial
$\sum_{\alpha\in {\cal M}}G_\alpha X^\alpha$ ; it lies in
$\Q(G)[X_0, \ldots, X_n]$.

We consider now the matrix
{\[
\Sigma_i=\begin{bmatrix}
  Y_0&\cdots & Y_n\\
  \partial G/\partial X_0 & \cdots&\partial G/\partial X_n\\
  \rule[2pt]{1cm}{.5pt}& A_1 & \rule[2pt]{1cm}{.5pt}\\
  &\vdots&\\
  \rule[1pt]{1cm}{.5pt} & A_i& \rule[2pt]{1cm}{.5pt}\\
\end{bmatrix}
\]} and the algebraic set
$\mathscr{S}_{i}\subset \C^{n+1}\times\C^{n+1}\times \C^N\times
\C^{(n+1)i}$ defined by $\rank(\Sigma_i)\leq i+1$.

Let $\sigma_1, \ldots, \sigma_L$ be the sequence of
$(i+1, i+1)$-minors of the submatrix $\Sigma_i$ obtained by removing
the line containing partial derivatives of $G$ or the line $A_j$ for
$1\leq j \leq L$ such that
$\mathscr{S}_i\cap {\cal O}(\sigma_\ell)\neq \emptyset$.  For
$1\leq \ell \leq L$, we denote by
$S_{1, \ell}, \ldots, S_{n-i-1, \ell}$ the $(i+2, i+2)$-minors of
$\Sigma_i$ obtained by selecting the rows and columns used to compute
$\sigma_\ell$ and adding the missing row and column from $\Sigma_i$.
We denote by $\mathbf{S}_\ell$ the sequence
$S_{1, \ell}, \ldots, S_{n-i-1, \ell}$.

Finally, we define the set
$\mathscr{T}\subset \C^N\times \C^{(n+1) i}$ as the complementary of
the set of points $(g,a=(a_1, \ldots, a_i))\in \C^N\times \C^{(n+1)i}$
such that
\begin{itemize}
\item the coefficients of $X_rX_s^{D-1}$ in $G$ for $1\leq r,s\leq n$
  with $r\neq s$ are not zero;
\item $(g,a)$ lies in the non-empty open set
  $\mathscr{O}$ defined in Lemma~\ref{lemma:Si};
\item ${\rm Span}(a_1, \ldots, a_i)$ has dimension $i$ and none of the
  entries of $A_r$ is $0$ (for $1\leq r\leq i$).
\end{itemize}
Note that $\mathscr{T}$ is Zariski closed in
$\C^N\times \C^{(n+1) i}$.  Finally, we denote by $\mathscr{S}'$ the
union of $\sing(\mathscr{S}_i)$, the set
$\C^{n+1}\times \C^{n+1}\times \mathscr{T}$ and the subset of points
$\mathscr{S}_i$ such that their $Y$-coordinates are all $0$.

Up to renumbering the sequence of couples
$(\mathbf{S}_\ell, \sigma_\ell)_{1\leq \ell \leq L}$ we assume that
the set of indices $\ell$ such that
$(\mathscr{S}_i\setminus\mathscr{S}')\cap {\cal O}(\sigma_\ell)\neq
\emptyset$ is $\{1, \ldots, L'\}$ (for $L'\leq L$). 

\begin{lemma}\label{lemma:atlas:S}
  The sequence $(\mathbf{S}_\ell, \sigma_\ell)_{1\leq \ell \leq L'}$
  is an atlas for the couple $(\mathscr{S}_i, \mathscr{S}')$. Besides,
  the truncated Jacobian matrix of $\mathbf{S}_\ell$ obtained by
  considering the partial derivatives w.r.t the entries of
  $A_1, \ldots, A_i$ and the coefficients of $G$ has full rank over
  ${\cal O}(\sigma_\ell)$.  Moreover, there exists a non-empty Zariski
  open set $\mathscr{O}''$ such that for all $(g,a)\in \mathscr{O}''$,
  $(\mathbf{S}_\ell, \sigma_\ell)_{1\leq \ell \leq L'}$ is an atlas of
  the couple $(S_i(g,a), \sing(S_i(g,a))$.
\end{lemma}

\begin{proof}
  Take $(x,y,g,a)$ in $\mathscr{S}_i\setminus \mathscr{S}'$.  Since
  $(g,a)\notin \mathscr{S}$, $(g, a)\notin \mathscr{O}$ and
  $(x,y)\notin \sing(S_i(g,a))$. We deduce that $\Sigma_i$ has rank
  $i+1$ at $(x,y,g,a)$. Then, either
  $\dim({\rm Span}(a_1, \ldots, a_i, y))=i+1$ or
  $y \in {\rm Span}(a_1, \ldots, a_i)$ while
  $\grad_{x,y,a,g}(G)\notin {\rm Span}(a_1, \ldots, a_i)$ (because
  $\Sigma_i$ has rank $i+1$ at $(x,y,g,a)$). Since
  $(y_0, \ldots, y_n)\neq 0$ (because $(x,y,g,a)\notin \mathscr{S}$),
  we deduce that there exists $1\leq r \leq i$ such that
  ${\rm Span}(a_1, \ldots, a_{r-1}, a_{r+1}, \ldots, a_i, y)={\rm
    Span}(a_1, \ldots, a_i)$ and we deduce that
\begin{equation*}
%  \label{eq:4}
  \dim({\rm Span}(a_1, \ldots, a_{r-1}, a_{r+1}, \ldots, a_i,\grad_x(g), y))=i+1.       
\end{equation*}
This implies that one of the $(i+1, i+1)$-minor $\sigma_\ell$ of
$\Sigma_i$ does not vanish at $(x, y, g, a)$. Elementary linear algebra
shows that
$\mathscr{S}_i\cap {\cal O}(\sigma_\ell)\setminus \mathscr{S}'$
coincides with $Z(\mathbf{S}_\ell)$ over
${\cal O}(\sigma_\ell)\setminus\mathscr{S}'$. Thus, we have
established properties ${\sf P}_1$ and ${\sf P}_2$. The covering property ${\sf P}_4$
is immediate and follows also from the above discussion.

It remains to prove property ${\sf P}_3$, i.e. $\jac(\mathbf{S}_\ell)$ has
maximal rank at any point of
$\mathscr{S}_i\cap {\cal O}(\sigma_\ell)\setminus \mathscr{S}'$.
Assume first that $\sigma_\ell$ is a $(i+1, i+1)$-minor obtained from
removing the partial derivatives of $G$ from $\Sigma_i$. Without loss of
generality, we may also assume that it is obtained by selecting the
first $i+1$ columns of $\Sigma_i$.  Then, polynomials in
$\mathbf{S}_\ell$ can be written as
$\sigma_\ell A_{r, i+1}+\rho_{r, \ell}$ for $i+1\leq r\leq n$ where
$\rho_{r,\ell}$ has degree $0$ in $A_r$. That implies that one can
extract a diagonal matrix with $\sigma_\ell$ on the diagonal from
$\jac(\mathbf{S}_\ell)$ which, of course, has maximal rank over
${\cal O}(\sigma_\ell)$.

When $\sigma_\ell$ is obtained by removing one of the line $A_r$
(e.g. $A_i$) a more involved but similar conclusion can be made. Since
we work over the complementary of $\mathscr{S}'$, there exists
$0\leq r \leq n$ such that the $X_r$-coordinate of $x$ is not
$0$. Extracting the submatrix of $\jac(\mathbf{S}_\ell)$ corresponding
to the partial derivatives with respect to the coefficients of $G$ of
the monomials $X_r^D$ and $X_sX_r^{D-1}$ yields a diagonal matrix with
a power of the $X_r$-coordinate of $x$ multiplied by $\sigma_\ell$ on
the diagonal. These are non-zero over
${\cal O}(\sigma_\ell)\setminus\mathscr{S}'$. 

The rank property of the truncated Jacobian matrix of
$\mathbf{S}_\ell$ is an immediate consequence of the above discussion.
Details on the proof of the specialization property of the atlas
$(\mathbf{S}_\ell, \sigma_\ell)$ are left to the reader; we mention
that it is a direct consequence of specialization properties of minors
with polynomial entries and Lemma~\ref{lemma:Si}.
\end{proof}

\subsection{A map and its regularity at the origin} 
Let $m_j$, $m'_k$ and $\mathbf{H}_{j,k}$ be the polynomials introduced
in the paragraph on local equations for $\conormal_V$ and
$\sigma_{\ell}$, $\mathbf{S}_{\ell}$ be the $(i+1, i+1)$-minor and
$(i+2, i+2)$-minors of $\Sigma$ introduced in the paragraph on local
equations of $S_i(g,a)$. Consider the Zariski open set
$\mathscr{U}_{j, k, \ell}\subset \C^n \times \C^n \times \C^N\times
\C^{ni}$ defined by
$$
m_jm'_k\neq 0, \quad (X_0, \ldots, X_n)\neq \mathbf{0}\quad (Y_0,
\ldots, Y_n)\neq \mathbf{0}, \quad \sigma_{\ell}\neq 0
$$
and the inequations defining the complement $\mathscr{S}'$. 
We define now the following map:
$$
  \bphi_{j, k, \ell}:z\in \mathscr{U}_{j, k, \ell}\to (\mathbf{H}_{j, k}(z),
  \mathbf{S}_{\ell}(z)).
$$
Observe that
$\bphi_{j,k,\ell}^{-1}(\mathbf{0})\subset \conormal_V\cap
\mathscr{S}_i$.

\begin{lemma}\label{lemma:maps:regularity}
  The map $\bphi_{j, k, \ell}$ is regular
  at $\mathbf{0}$.
\end{lemma}
\begin{proof}
  Since $j, k$ and $\ell$ are fixed in the sequel, we omit them as
  subscripts. Observe that $\jac(\mathbf{H}, \mathbf{S})$ has the
  following shape
  \[
  \mathbf{J}_{\bphi}=\left [
    \begin{array}{cccc}
      \jac_\X(\mathbf{H})& \mathbf{0}\\
      \jac_\X(\mathbf{S})& \Delta \\
    \end{array}
  \right ]
  \]
  where the last columns correspond to the partial derivates with
  respect to the entries of $A_1, \ldots, A_i$ and $G$. By
  Lemma~\ref{lemma:atlas:conormal}, $(\mathbf{H}_{j, k},\penalty-1000 m_jm'_k)$
  satisfies properties ${\sf P}_1, {\sf P}_2$ and ${\sf P}_3$. This implies that it
  has maximal rank at any point in
  $\bphi^{-1}(\mathbf{0})\subset \mathscr{U}$. By
  Lemma~\ref{lemma:atlas:S}, $\Delta$ has maximal rank at any point of
  $\bphi^{-1}(\mathbf{0})$. We deduce that $\mathbf{J}_{\bphi}$ has
  maximal rank at any point of $\bphi^{-1}(\mathbf{0})$ and our
  conclusion follows.
\end{proof}

In the sequel, for $(g,a)\in \C^N\times \C^{(n+1)i}$, we denote by
$\bphi_{j,k,\ell}^{(g,a)}$ the restricted map
$(x, y)\to \bphi_{j,k,\ell}(x,y,g,a)$. Applying Thom's weak
transversality Theorem (see \cite[Sec 4.2]{SaSc13}) to
$\bphi_{j,k,\ell}$ shows that there exists a non-empty Zariski open
set $\mathscr{O}'''_{j,k,\ell}\subset \C^N\times \C^{(n+1)i}$ such
that for all $(g,a)\in \mathscr{O}'''_{j,k,\ell}$, the restricted map
$\bphi_{j,k,\ell}^{(g,a)}$ is regular at $\mathbf{0}$. Letting
$\mathscr{O}'''$ be the intersection of all these non-empty Zariski
open sets $\mathscr{O}'''_{j,k,\ell}$ leads to the following result.
\begin{lemma}\label{lemma:regular:spec}
  There exists a non-empty Zariski open set
  $\mathscr{O}'''\subset\C^N\times \C^{(n+1)i}$ such that for any
  $(j,k,\ell)$ and $(g,a)\in \mathscr{O}'$, the restricted map
  $\bphi_{j,k,\ell}^{(g,a)}$ is regular at $\mathbf{0}$.
\end{lemma}
\subsection{Transversality of the intersection} 

Let $\mathscr{O}$ be the intersection of the non-empty Zariski open
sets $\mathscr{O}''$ and $\mathscr{O}'''$ defined in
Lemma~\ref{lemma:atlas:S} and Lemma~\ref{lemma:regular:spec}. Take
$(g,a)\in \mathscr{O}$ and $Z_{(g,a)}$ be the Zariski closure of
$\bigcup_{j, k, \ell}{\bphi_{j,k,\ell}^{(g,a)}}^{-1}(\mathbf{0})$.
Recall that we need to prove the transversality of
$\conormal_V\cap S_i(g,a)$ at any point outside
$\mathscr{X}\cup\mathscr{Y}$.
Let $\balpha_1=( \mathbf{H}_{j,k}, m_jm'_k)$ be the atlas of
$(\conormal_V, \sing(\conormal_V))$ defined in
Lemma~\ref{lemma:atlas:conormal} and
$\balpha_2=(\mathbf{S}_\ell, \sigma_\ell)$ be the atlas of
$(S_i(g,a), \sing(S_i(g,a)))$ defined in Lemma~\ref{lemma:atlas:S}.
We start by proving that the Zariski closure of
$\conormal_V\cap S_i(g,a)\setminus(\mathscr{X}\cup\mathscr{Y})$ equals
$Z_{g,a}$. The inclusion
$Z_{(g,a)}\subset \conormal_V\cap
S_i(g,a)\setminus(\mathscr{X}\cup\mathscr{Y})$
is immediate since all points in
${\bphi_{j,k,\ell}^{(g,a)}}^{-1}(\mathbf{0})\subset Z(\mathbf{H}_{j,k},
\mathbf{S}_{k,\ell})\cap {\cal O}(m_j m'_k\sigma_\ell)$
and
$Z(\mathbf{H}_{j,k}, \mathbf{S}_{k,\ell})\cap {\cal O}(m_j
m'_k\sigma_\ell) = \conormal_V\cap S_i(g,a)\cap {\cal O}(m_j
m'_k\sigma_\ell)$ (property ${\sf P}_2$).
We prove now the reverse inclusion.  It is sufficient to prove that
for any irreducible component $Z$ of the Zariski closure of
$\conormal_V\cap S_i(g,a)\setminus(\mathscr{X}\cup\mathscr{Y})$, there exists
a triple $(j,k,\ell)$ and a Zariski closed subset $F\subsetneq Z$ such
that $Z\setminus F\subset {\bphi_{j,k,\ell}^{(g,a)}}^{-1}(\mathbf{0})$.  Since
$Z$ is an irreducible component of the Zariski closure of
$\conormal_V\cap S_i(g,a)\setminus (\mathscr{X}\cup\mathscr{Y})$,
there exists $(x, y)\in Z$ such that
$(x, y)\notin \mathscr{X}\cup\mathscr{Y}$. Let
$F=Z\cap (\mathscr{X}\cup \mathscr{Y})$.
Now, take $(x, y)\in Z\setminus F$. By property ${\sf P}_4$ applied to $\balpha_1$,
that implies that there exists $j$ and $k$ such that
$x\in Z(\mathbf{H}_{j,k})\cap {\cal O}(m_jm'_k)$.  Besides,
$(y_0, \ldots, y_n)\neq \mathbf{0}$ since $(x, y)\notin F$. This
latter property implies that there exists $\ell$ such that
$\sigma_\ell(x,y)\neq 0$. 
Finally, we have established that
$(Z\setminus F)\cap {\cal O}(m_jm'_k\sigma_\ell)$ is not empty for some
$(j,k,\ell)$.  Property ${\sf P}_2$ applied to $\balpha_1$ and $\balpha_2$
imply that $(x, y)$ lies in $Z(\mathbf{H}_{j,k})$ and
$Z(\mathbf{S}_{\ell})$. We deduce that
$(x, y)\in {\bphi_{j,k,\ell}^{(g,a)}}^{-1}(\mathbf{0})$ which
implies that $Z\setminus F\subset {\bphi_{j,k,\ell}^{(g,a)}}^{-1}(\mathbf{0})$
as requested.

{\bf Property $({\sf T}_1)$.}  Consider an irreducible component $Z$
of the Zariski closure of
$\conormal_V\cap S_i(g,a)\setminus (\mathscr{X}\cup \mathscr{Y})$. The above
discussion implies that $Z$ is an irreducible component of $Z_{(g,a)}$
and that there exists $j,k,\ell$ such that
$Z\cap {\cal O}(m_jm'_k\sigma_\ell)$ is not empty.

{\bf Property $({\sf T}_2)$.} Recall that $(g,a)\in \mathscr{O}'$ and
let $Z$ be an irreducible component of $\conormal_V\cap S_i(g,a)$. We
already proved that $Z$ there exists $(j,k,\ell)$ such that
$Z\cap {\cal O}(m_jm'_k\sigma_\ell)$ is not empty and meets
${\bphi_{j, k, \ell}^{(g,a)}}^{-1}(\mathbf{0})$. By
Lemma~\ref{lemma:regular:spec}, the restricted map
$\bphi_{j, k, \ell}^{(g,a)}$ is regular at $\mathbf{0}$. Then, the
jacobian matrix associated to $\mathbf{H}_{j, k}, \mathbf{S}_{\ell}$
has maximal rank at any point of
$Z\cap {\bphi_{j, k, \ell}^{(g,a)}}^{-1}(\mathbf{0})$, which
concludes the proof.

%%% Local Variables:
%%% mode: latex
%%% TeX-master: "main"
%%% End:

%%% Local Variables:
%%% mode: latex
%%% TeX-master: "main"
%%% End:

\section{Non-generic function}
\label{sec:nongeneric}
We show in this section that the bounds in Theorem
\ref{thm:degreeModPolar} hold under milder conditions than the
genericity of the coefficients of $g$.  Consider a
$d$-equidimensional algebraic set $V\subset \C^n$ whose projective closure is
smooth, a set of generators
$f_1,\ldots, f_p$ of $I(V)$, and $g\in\Q[X_1,\ldots, X_n]$ of degree
$D$. Let $a\in\C^{ni}$, $g\in\Q[X_1,\ldots, X_n]$ and $I_{\sf crit}(g,a)$ be
the ideal generated by $f_1,\ldots, f_p$ and the $(n-d+i+1)$-minors of the
matrix
{
$$\begin{bmatrix}
\rule[2pt]{1cm}{.5pt}&\jac(\f)&\rule[2pt]{1cm}{.5pt}\\
\partial g/\partial X_1 & \cdots&\partial g/\partial X_n\\
\rule[2pt]{1cm}{.5pt}& a_1 & \rule[2pt]{1cm}{.5pt}\\
&\vdots&\\
\rule[1pt]{1cm}{.5pt} & a_{i}& \rule[2pt]{1cm}{.5pt}\\
\end{bmatrix}.$$}
\begin{proposition}\label{prop:nongeneric} Let $i\in\{0,\ldots, d\}$ and
  $a\in\C^{ni}$.
  Assume that the ideal $I_{\sf
    crit}(g,a)$ is radical and $W(g,a)$ is empty or
  $(i-1)$
  equidimensional. Then there exists a non-empty Zariski open subset
  $\mathscr
  O\subset \C^{i n}$ such that the following holds. For any $a=(a_1,
  \ldots, a_i)\in\mathscr O$, the degree of $W((g,\mathbf
  a),V)$ is bounded above by the bounds in Theorem
  \ref{thm:degreeModPolar}.
\end{proposition}
\begin{lemma}\label{lem:existencePoly}
  Let $Q\in\C[T_1,\ldots, T_\ell]$ be a nonzero multivariate
  polynomial, and $(t_1,\ldots, t_\ell)\in\C^\ell$ be such that
  $Q(t_1,\ldots,t_\ell) = 0$. Then there exist univariate polynomials
  $u_1,\ldots, u_\ell\in\C[\mathfrak e]$ such that for all
  $i\in\{1,\ldots,\ell\}$, $u_i(0)=t_i$ and
  $Q(u_1(\mathfrak e),\ldots, u_\ell(\mathfrak e))\in\C[\mathfrak e]$
  is not identically zero.
\end{lemma}

\begin{proof}
  We prove the existence of $u_1,\ldots, u_\ell$ of the form
  $u_i(\mathfrak e)=t_i+ s_i\mathfrak e$, where $s_i\in\C$ for all
  $i\in\{1,\ldots,\ell\}$.  Let $\mathbf t$ and $\mathbf s$ be
  shorthands for $(t_1,\ldots, t_\ell)$ and $(s_1,\ldots, s_\ell)$.
  Using Taylor's expansion, we write
  $Q(\mathbf t+\mathfrak e \mathbf s)= \mathfrak e \partial Q(\mathbf
  t)(\mathbf s)+ \mathfrak e^2 \partial^2 Q(\mathbf t)(\mathbf
  s,\mathbf s)/2 + \ldots +\mathfrak e^{\deg(Q)} \partial^{\deg(Q)} Q
  (\mathbf t)(\mathbf s,\ldots, \mathbf s)/\deg(Q)!$.
  Since $Q\neq 0$, at least one of its derivatives is not zero at
  $\mathbf t$. Let $k$ be the smallest integer such that
  $\mathbf u\mapsto \partial^k Q(\mathbf t)(\mathbf u,\ldots, \mathbf
  u)$
  is not the zero map. Finally, let $\mathbf s$ be such that
  $\partial^k Q(\mathbf t)(\mathbf s,\ldots, \mathbf s)\neq 0$.
  Hence, we have
  $Q(\mathbf t+\mathfrak e\mathbf s)-\mathfrak e^k \partial^k
  Q(\mathbf t)(\mathbf s,\ldots, \mathbf s)/k! = O(\lvert \mathfrak
  e^{k+1}\rvert)$.
  Consequently, $Q(\mathbf t+\mathfrak e\mathbf s)$ cannot be
  identically zero, as this would imply
  $\mathfrak e^k = O(\lvert\mathfrak e^{k+1}\rvert)$.
\end{proof}

\begin{proof}[of Proposition \ref{prop:nongeneric}] The proof is a classical deformation argument similar to the one used in
  \cite{nie2009algebraic}. 
%  Let $a\in\C^{ni}$ be a
%  sequence of $a$ generic vectors in $\C^n$.  
  Further we assume that $W(g,a)$ is not empty (else the result is
  immediate).  By Theorem \ref{thm:degreeModPolar}, there exists a
  polynomial $Q$ in $N+ni=\binom{n+D}{n}+ni$ variables whose zero-set
  encode the pairs $(g,a)$ for which the bounds are not satisfied.  By
  Lemma \ref{lem:existencePoly}, there exists
  $(\mathfrak g, \mathfrak a)\in\Q[\mathfrak e][X_1,\ldots, X_n]\times
  \Q[\mathfrak e]^{ni}$
  such that their evaluation at $\mathfrak e = 0$ is $(g,a)$ and the
  evaluation of $Q$ at $(\mathfrak g, \mathfrak a)$ (seen as an
  element in $\Q[\mathfrak e]^{N+ni}$) is nonzero. For
  $\varepsilon\in\C$, we let $(g_\varepsilon, a_\varepsilon)$ denote
  the evaluation of $\mathfrak g$ and $\mathfrak a$ at $\mathfrak e=\varepsilon$. For
  $i\in\{1,\ldots, d\}$, the set of affine spaces in $\C^n$ of
  codimension $i-1$ can be identified with a dense open subset of the
  Grassmaniann of $(n-i+2)$-dimensional vector spaces in $\C^{n+1}$.
  Since $W((g,a),V)$ is $(i-1)$-equidimensional, there exists a dense
  open subset $\mathscr O$ of this Grassmanian such that for any $E$
  in $\mathscr O$, the intersection $W( (g,a), V)\cap E$ is
  transverse, finite and its cardinality equals the degree of
  $W( (g,\mathbf a_i),V)$.  Let $\mathbf x\in\C^n$ be a point in this
  intersection.  Let
  $v_1,\ldots, v_n\in\C[X_1,\ldots, X_n,\mathfrak e]$ be polynomials
  satisfying the following assumptions: $v_1,\ldots, v_n$ is a regular
  sequence, for every $\varepsilon\in\C$ their evaluations at
  $\mathfrak e=\varepsilon$ vanish on
  $W( (g_\varepsilon,a_\varepsilon), V)\cap E$, and the jacobian
  matrix
  $\jac(v_1(X_1,\ldots,X_n,\penalty-1000 0),\ldots, v_n(X_1,\ldots,
  X_n, 0))$
  is invertible at $\mathbf x$ (since $I_{\sf crit}(g,a)$ is
  radical). In order to obtain such polynomials, we
  consider $n$ generic linear combinations of the equations defining
  $W( (\mathfrak g, \mathfrak a), V)\cap E$. Then the holomorphic implicit mapping theorem
  \cite[Thm.  8.6]{kaup1983holomorphic} states that for
  $\mathbf x_0\in W( (g, a), V)\cap E$ there exist open neighborhoods
  (for the Euclidean topology) $0\in U_1\subset \C$,
  $\mathbf x_0\in U_2$ such that there is a holomorphic map
  $\varepsilon \mapsto \{\mathbf x\in U_2 \mid v_1(\mathbf x,
  \varepsilon) = \cdots =v_n(\mathbf x, \varepsilon) = 0\}$
  on $U_1$.  In particular this map is continuous, which implies that
  for $\varepsilon\in\C$ with sufficiently small complex modulus, the
  cardinality of $W( (g_\varepsilon,a_\varepsilon), V)\cap E$ is
  bounded below by the degree of $W( (g, a), V)$. Since this is true
  for any $E$ in the Zariski dense open subset $\mathscr O$ of affine
  subsets, the cardinality of
  $W( (g_\varepsilon, a_\varepsilon), V)\cap E$ equals its
  degree. Finally, as $Q$ is not identically zero on the coefficients
  of $(\mathfrak g, \mathfrak a)$, for $\varepsilon_0$ with
  sufficiently small modulus, the evaluation of $Q$ at the
  coefficients of $(g_{\varepsilon_0},a_{\varepsilon_0})$ is
  nonzero. Consequently, the bounds in Theorem
  \ref{thm:degreeModPolar} hold for
  $W( (g_{\varepsilon_0}, a_{\varepsilon_0}), V)$ and hence they also
  hold for $W( (g, a), V)$.  \end{proof}

%%% Local Variables:
%%% mode: latex
%%% TeX-master: "main"
%%% End:

\section{Algorithms}
\label{sec:algo}
{\bf Terminology and computational model.} In this section, we consider \emph{bounded
error probabilistic algorithms}. These
algorithms are probabilistic random-access stored-program machines whose
probability of success is bounded from above by an \emph{a priori} bound. 
It is the same computational model as in \cite{giusti2001grobner}.
Complexity bounds count the number of arithmetic operations
($+$, $-$, $\times$, $/$) in $\Q$. 

A lifting fiber is a data structure giving an exact representation of
an equidimensional algebraic set.  We recall below its definition and
we refer to \cite[Sec. 3.4]{giusti2001grobner} for more details.
\begin{definition}\cite[Def. 4]{giusti2001grobner}\label{def:liftingfiber}
Let $V\subset \C^n$ be a $d$-equi\-dimen\-sional variety defined over $\Q$
(\emph{i.e.} $Z_\C(I_\Q(V))=V$). A \emph{lifting fiber} for $V$
 is a tuple $\mathscr L = (\mathbf G, M, \mathbf z, u, Q, \mathbf v)$:
\begin{itemize}
  \item A \emph{lifting system} $\mathbf H = (h_1,\ldots, h_{n-d})\in\Q[X_1,\ldots, X_n]$, such that
    $h_1,\ldots, h_{n-d}$ is a reduced regular sequence and $V\subset Z(\mathbf
    H)$.
  \item A $n\times n$ invertible matrix $M$ with entries in $\Q$ such that
    the coordinates $Y=M^{-1} X$ are in \emph{Noether position} w.r.t. $V$;
  \item A \emph{rational lifting point} $\mathbf z=(z_1,\ldots, z_{d})\in\Q^{d}$;
  \item A \emph{primitive element} $u: \C^n\rightarrow \C$, which is a linear
    form with rational coefficients having distinct values at all points of the
    finite set $V^{(\mathbf z)} = V\cap \{Y_1 - z_1 = \dots = Y_d-z_d = 0\}\subset
    \C^n$;
  \item A \emph{polynomial} $Q\in\Q[T]$ of minimal degree vanishing at all
    points of $u(V^{(\mathbf z)})$;
  \item \emph{univariate polynomials}
    $\mathbf v=(v_{d+1},\ldots, v_n)\in\Q[T]^{n-d}$ of degree less
    than $\deg(Q)$ such that
    {$$\begin{array}{c}
                                         Y_1-z_1 = \dots  = Y_{d}-z_{d}=0\\
                                         Y_{d+1} - v_{d+1}(T) = \dots = Y_n-v_n(T)=0, 
                                         Q(T) = 0
             \end{array}$$}
           is a rational parametrization of $V^{(\mathbf z)}$ by the roots of $Q$.
\end{itemize}
The sequence $(M, u, Q, \mathbf v)$ is called a \emph{geometric
  resolution} of $V$.
\end{definition}
Computing a lifting fiber can be achieved in a probabilistic way with
the Kronecker solver \cite{durvye2008concise}. We assume that we know a probabilistic algorithm \textsc{PolarVar}
which takes as input $d\in\N$, a lifting fiber of a
$d$-equidimensional variety $V\subset\C^n$ and a sequence
$\mathbf a =(a_1,\ldots, a_d)\in\Q^{d\times n}$; it returns a
geometric resolution of the $0$-dimensional polar variety
$W(\mathbf a_1,V)$ or ``fail''. We use also the routine
\textsc{ChangePrimitiveElement} \cite[Algo. 6]{giusti2001grobner}.

In \cite{bank2013degeneracy}, the authors propose an algorithm which
takes as input a reduced regular sequence $f_1,\ldots, f_{n-d}$
defining a $d$-equidimensional algebraic set $V\subset\C^n$, a matrix
$\mathbf F$ whose entries are multivariate polynomials, and a sequence
$a=(a_1\ldots, a_d)$ of vectors in $\Q^n$. It returns lifting
fibers for the associated \emph{degeneracy loci}. If $\mathbf F$ turns
out to be the jacobian matrix of the regular sequence defining the
variety, then these degeneracy loci are the classical polar varieties
$W(a, V)$, see \cite[Section 5.1]{bank2013degeneracy}. This
algorithm works in two steps: it computes first a lifting fiber for
$V$; then, from this lifting fiber and from the matrix
$a$, it computes lifting fibers for the degeneracy loci.  In
the case of polar varieties, the complexity of the second step is
bounded by $L(nD_{\max})^{O(1)}\delta^2$, where $\delta$ is the
maximum of the degrees of the polar varieties $W(\mathbf a_i, V)$ (where
$\mathbf a_i=(a_1,\ldots, a_i)$),
$D_{\max}$ is the maximum of the degrees of $f_1,\ldots, f_{n-d}$, and
$L$ is the size of an essentially division-free straight line program
for evaluating $f_1,\ldots, f_{n-d}$.

Let $a=(a_1,\ldots, a_d)$ be a sequence of $d$ vectors in $\Q^n$. We
construct another sequence $a'=(e_{n+1}, a_1',\ldots, a_d')$ of vectors
in $\Q^{n+1}$ defined by the $(d+1)\times(n+1)$ coefficient
matrix
{$$A'=\left[\begin{array}{ccc|c}
  &\mathbf 0& & 1\\
  \hline
  \rule[2pt]{1cm}{.5pt}&a_1&\rule[2pt]{1cm}{.5pt}&\\
  \vdots&\vdots&\vdots&\mathbf 0\\
  \rule[2pt]{1cm}{.5pt}&a_d&\rule[2pt]{1cm}{.5pt}&
\end{array}\right].
$$}

\begin{lemma}\label{lem:projcrit} Let $\Pi_n:\C^{n+1}\rightarrow\C^n$ be the
  projection on the $n$ first coordinates, and $f_1,\ldots,
  f_p\in\Q[X_1,\ldots, X_n]$ be polynomials defining a reduced smooth
  $d$-equidimensional variety and $g\in\Q[X_1,\ldots, X_n]$ be a polynomial.
  Then for any $\mathbf a\in\Q^{d\times n}$ and for $i\in\{0,\ldots,
  d\}$, the modified polar variety $W((g,a), Z(f_1,\ldots, f_p))$ equals
  $\Pi_n(W(a'_{i+1}, Z(f_1,\ldots, f_p, g-X_{n+1}))).$ \end{lemma}

\begin{proof}
  Set $V=Z(f_1,\ldots, f_p)\subset \C^n$ and $V' = Z(f_1,\ldots,
  f_p,\penalty-1000
  g-X_{n+1})\subset\C^{n+1}$.
  Direct computations show that if $V$ is smooth, then so is $V'$.
  The modified polar variety
  $W( (g,\mathbf a_i), V)$ is defined by the set of points in $V$ at which
{
$$\rank \left[\begin{array}{ccc}  
    \rule[2pt]{1cm}{.5pt}& \jac(\f)&\rule[2pt]{1cm}{.5pt}\\
    \rule[2pt]{1cm}{.5pt}&\nabla g&\rule[2pt]{1cm}{.5pt}\\
   \rule[2pt]{1cm}{.5pt} &a_1&\rule[2pt]{1cm}{.5pt}\\
   \vdots&\vdots&\vdots\\
   \rule[2pt]{1cm}{.5pt} &a_1&\rule[2pt]{1cm}{.5pt}\\
\end{array}\right]\leq n-d+i.$$}
Direct computations show that the corresponding matrix for $W(\mathbf
a'_i,\penalty-1000 V')$
has the same rank at any point $(x,g(x))$ where $x\in V$.
\end{proof}
{
\begin{algorithm}\label{algo:georescritpoints}\caption{\textsc{CritPoints}}
  \SetKwInOut{Input}{Input}
  \SetKwInOut{Output}{Output}
  \Input{{\begin{itemize}
  \item A lifting fiber $(\mathbf H, M, \mathbf z, u, Q,
 \mathbf v)$ for a smooth $d$-equidimensional
    variety $V\subset\C^n$
  \item $g\in\Q[X_1,\ldots, X_n]$ and  $a = (a_1,\ldots, a_d)\in\Q^{d\times n}$
  \item A primitive element $u_{\sf crit}$ for $W(\mathbf a_1, V)$
\end{itemize}}}

{$M'\gets\left[\begin{array}{ccc|c}
     &&&\\
     &M&&\mathbf 0\\
     &&&\\\hline
   &\mathbf 0&&1\end{array}\right]$\;

$\mathscr L'\gets(\mathbf H\cup\{g-X_{n+1}\}, M', (z_1,\ldots, z_n, g(z_1,\ldots,
z_n)), \penalty-1000 u, Q, (v_1,\ldots, v_n, g\circ (v_1(T),\ldots, v_n(T))\bmod Q(T)))$\;

$a'\gets \text{sequence of rows of } \left[\begin{array}{ccc|c}
  &\mathbf 0& & 1\\
  \hline
  \rule[2pt]{1cm}{.5pt}&a_1&\rule[2pt]{1cm}{.5pt}&\\
  \vdots&\vdots&\vdots&\mathbf 0\\
  \rule[2pt]{1cm}{.5pt}&a_d&\rule[2pt]{1cm}{.5pt}&
\end{array}\right]
$\;

$\mathscr L^{(2)}\gets \textsc{PolarVar}(d,\mathscr L', a')$\text{ or return ``fail''}\;

$(\mathbf H', M^{(2)}, \mathbf z', u_{\sf crit}, Q', \mathbf
v')\gets\textsc{ChangePrimitiveElement}(\mathscr L^{(2)}, u_{\sf crit}\circ
M^{(2)})$\;

$\mathbf v^{(2)}\gets (M^{(2)})^{-1}\cdot (v'_1,\ldots, v'_{n+1})^T$\;

return $({\sf Id}_n,u_{\sf crit}\circ (M^{(2)})^{-1}, Q', (v'_1,\ldots, v'_n))$\;}
\end{algorithm}
}
\begin{theorem} 
  Let $\mathscr L=(\mathbf H, M, \mathbf z, u, Q, \mathbf v)$ be a
  lifting fiber for a $d$-equidimensional algebraic set $V$,
  $g\in \Q[X_1,\ldots,X_n]$ be a polynomial of degree $D\geq 2$ and
  $a\in\C^{ni}$. Assume that $V, g$ and $a$ satisfy the same
  assumptions as in Proposition \ref{prop:nongeneric}. Let $D_{\max}$
  be the maximum of the degrees of $h_1,\ldots, h_{n-d}, g$, (where
  $\mathbf H=(h_1,\ldots, h_{n-d})$) and $u$ be a primitive element
  for $W(g,V)$. Assume that the evaluation map
  $x\mapsto (h_1(x),\ldots,h_{n-d}(x), g(x))$ is represented by an
  essentially division-free straight-line program of size $L$.
  Algorithm \ref{algo:georescritpoints} with input
  $(\mathscr L, \mathbf a, u)$ computes a geometric resolution
  of the set $W(g,V)$ or it returns ``fail''. Using the algorithm in
  \cite{bank2013degeneracy} for \textsc{PolarVar}, it requires at most
  $(n D_{\max})^{O(1)}\widetilde O(L \Delta^2)$ operations in $\Q$,
  where $\Delta=\sum_{j=0}^{d} \delta_{j+1}(V) (D-1)^{j-i}$.
\end{theorem}

\begin{proof}
  We prove first the correctness of the algorithm. Note that
  $\mathscr L'$ computed during Algorithm \ref{algo:georescritpoints}
  is a lifting fiber for $V' = \{( x, g(x))\mid x\in V\}$.  Assuming
  that \textsc{PolarVar} returns a lifting fiber for
  $W(\mathbf a'_1,V')$, Lemma~\ref{lem:projcrit} shows that the output
  of {\textsc{PolarVar}} is a lifting fiber of the pairs $(x,g(x))$
  for $x\in W(g, V)$. The last steps compute a geometric resolution of
  the projection on the $n$ first coordinates, which is $W(g, V)$.
  We prove now the complexity statement. The first step of
  Algorithm~\ref{algo:georescritpoints} does not cost any arithmetic
  operations. The second step requires $\widetilde O(L \deg(V))$ operations
  in $\Q$ for the modular composition using quasi-linear algorithms for multiplication and
  reduction. The evaluation of $g$ costs $L$ operations. The cost of
  the computation of $\mathbf a'$ is negligible. By
  \cite[Thm. 18]{bank2013degeneracy}, the call to \textsc{PolarVar}
  requires $L (p n d)^{O(1)} \delta'^2$, where $\delta'$ is the
  maximum of the degrees of the polar varieties of $V'$. By Lemma
  \ref{lem:projcrit}, the projection of $W(\mathbf a'_1, V')$ on the
  $n$ first coordinates is $W( g,V)$. By
  Theorem~\ref{thm:degreeModPolar}, Proposition~\ref{prop:nongeneric}
  and since $\deg(g)\geq 2$, we have $\deg(W( g, V))\leq \Delta$. Changing the
  primitive element costs $\widetilde O(n \Delta^2)$ by \cite[Lemma
  6]{giusti2001grobner}.  Finally, the linear algebra computations in
  the last step cost $O(n^2 \Delta)$ operations in $\Q$. Summing all
  these complexities proves the complexity statement. \end{proof}

%%% Local Variables:
%%% mode: latex
%%% TeX-master: "main"
%%% End:

\bibliographystyle{abbrv}
\bibliography{biblio}
\end{document}